\documentclass[12pt]{amsart}
\usepackage{graphicx}

\advance\textheight\topskip
\usepackage{amsmath,amsfonts,amsthm,amssymb,amscd}
\usepackage[mathscr]{eucal}
\allowdisplaybreaks
\usepackage{times}
\usepackage[curve,matrix,arrow]{xy}

\newcommand{\modM}{\mathscr{M}}
\newcommand{\modN}{\mathscr{N}}

\newcommand{\idem}{\boldsymbol{e}}

\newcommand{\bref}[1]{\textbf{\ref{#1}}}

\renewcommand{\geq}{\,{\geqslant}\,}
\renewcommand{\leq}{\,{\leqslant}\,}

\renewcommand{\le}{\,{\leqslant}\,}

\newcommand{\Ext}{\mathrm{Ext}_{\rule{0pt}{9.5pt}%
    \overline{\mathscr{LU}}_q}^1}
\newcommand{\ExtU}{\mathrm{Ext}_{\rule{0pt}{9.5pt}%
    \overline{\mathscr{U}}_q}^1}
\newcommand{\ExtUbul}{\mathrm{Ext}_{\rule{0pt}{9.5pt}%
    \overline{\mathscr{U}}_q}^{\bullet}}

\newcommand{\tensor}{\otimes}

\newcommand{\q}{\mathfrak{q}}

\newcommand{\catC}{\mathscr{C}}
\newcommand{\catCbar}{\overline{\mathscr{C}}}
\newcommand{\catCpl}{\mathscr{C}^{+}}
\newcommand{\catCmin}{\mathscr{C}^{-}}
\newcommand{\catSpl}{\mathscr{S}^{+}}
\newcommand{\catSmin}{\mathscr{S}^{-}}
\newcommand{\End}{\mathrm{End}}

\newcommand{\fusion}{%
  \mathop{{\otimes}\kern-7pt\raisebox{.6pt}{%
      \mbox{\footnotesize${\bullet}$}}}}

\newcommand{\UqSL}[1]{\mathscr{U}_{\q} s\ell(#1)}

\newcommand{\UresSL}[1]{\overline{\mathscr{U}}_{\q} s\ell(#1)}
\newcommand{\LUresSL}[1]{\mathscr{LU}_{\q} s\ell(#1)}

\newcommand{\ffrac}[2]{\mbox{\footnotesize$\displaystyle\frac{#1}{#2}$}}
\newcommand{\half}{%
  \mathchoice{\ffrac{1}{2}}{\frac{1}{2}}{\frac{1}{2}}{\frac{1}{2}}}

\newcommand{\LM}{{\mathcal{LM}}}
\newcommand{\WLM}{{\mathcal{WLM}}}

\newcommand{\repX}{\mathscr{X}}

\newcommand{\XX}{\mathscr{X}} 

\newcommand{\PP}{\mathscr{P}}

\newcommand{\voal}{\mathcal} 
\newcommand{\Vir}{\voal V_{p}}
\newcommand{\Virpq}{\voal V_{p,q}}
\newcommand{\Walg}{\voal W_{p}}
\newcommand{\modVir}{\mathcal{R}}
\newcommand{\modR}{\mathcal{R}}

\newcommand{\oC}{\mathbb{C}}

\newcommand{\oN}{\mathbb{N}}

\newcommand{\oZ}{\mathbb{Z}}

\newcommand{\one}{\boldsymbol{1}}

\newcommand{\toppr}{\mathsf{t}}
\newcommand{\botpr}{\mathsf{b}}
\newcommand{\leftpr}{\mathsf{l}}
\newcommand{\rightpr}{\mathsf{r}}

\newcommand{\stprp}{\mathsf{a}}

\newcommand{\cas}{\boldsymbol{C}}

\numberwithin{equation}{section}
\makeatletter
\@addtoreset{equation}{section}
\@addtoreset{subsubsection}{section}

\def\@secnumfont{\bfseries}
\def\subsubsection{\@startsection{subsubsection}{3}%
  \z@{.5\linespacing\@plus.7\linespacing}{-.5em}%
  {\normalfont\bfseries}}
\def\paragraph{\@startsection{paragraph}{4}%
  \z@\z@{-\fontdimen2\font}%
  \normalfont\bfseries}
\def\subparagraph{\@startsection{subparagraph}{5}%
  \z@\z@{-\fontdimen2\font}%
  \normalfont\bfseries}

\makeatother

\newcommand{\rme}{{\rm e}}
\newcommand{\algW}{\mathcal{W}}

\newcommand{\step}{\text{step}}

\swapnumbers
\newtheorem{Thm}[subsection]{Theorem}
\newtheorem{thm}[subsubsection]{Theorem}
\newtheorem{Lemma}[subsection]{Lemma}
\newtheorem{lemma}[subsubsection]{Lemma}

\newtheorem{prop}[subsubsection]{Proposition}

\theoremstyle{definition}

\newtheorem{rem}[subsubsection]{Remark}

\begin{document}
\title[Lusztig limit and fusion]{%
  \vspace*{-4\baselineskip}
  \mbox{}\hfill
  \\[\baselineskip] Lusztig limit of quantum sl(2) at root of unity
 and fusion of (1,p) Virasoro logarithmic minimal models}
\author{P.V.~Bushlanov, B.L.~Feigin, A.M.~Gainutdinov and I.Yu.~Tipunin}
\address{PVB:Moscow Institute of Physics and Technology, Dolgoprudny, Moskovskoe shosse 21a , Russia, 141700 }
\email{paulbush@mail.ru} 
\address{BLF:Higher School of Economics, Moscow, Russia and Landau institute for Theoretical Physics,
Chernogolovka, 142432, Russia}
\email{bfeigin@gmail.com}
\address{AMG:Tamm Theory Division, Lebedev Physics Institute, Leninski pr., 53,
Moscow, Russia, 119991}
\email{gainut@gmail.com} 
\address{IYuT:Tamm Theory Division, Lebedev Physics Institute, Leninski pr., 53,
Moscow, Russia, 119991}
\email{tipunin@gmail.com} 

\maketitle

\begin{abstract}
 We introduce a Kazhdan--Lusztig-dual quantum group for $(1,p)$
 Virasoro logarithmic minimal models as the Lusztig limit of the
 quantum $s\ell(2)$ at $p$\,th root of unity and show that this limit
 is a Hopf algebra. We calculate tensor products of irreducible and
 projective representations of the quantum group and show that these
 tensor products coincide with the fusion of irreducible and
 logarithmic modules in the $(1,p)$ Virasoro logarithmic minimal
 models.
\end{abstract}

\section{Introduction}
Logarithmic conformal field theories most naturally appear as a
scaling limit of two-dimensional nonlocal lattice models at a critical
point~\cite{[PRZ]} and in quantum chains with a nondiagonalizable
Hamiltonian~\cite{RS-Q}. Generally speaking, a conformal field
theory appearing at the limit depends on a way of taking the limit and
on chosen boundary conditions.  The recent
investigations~\cite{[PPR],[PR],[J.R]} argue that for proper choice of
boundary conditions the lattice models~\cite{[PRZ]} give in a scaling
limit logarithmic conformal models $\WLM(p,q)$ with the triplet
$\algW_{p,q}$-algebra of symmetry introduced in~\cite{[FGST3]} and
in~\cite{[K-first],[FHST]} for $q=1$. The chiral algebra in these
conformal models is an extension of the vacuum module of the Virasoro
algebra $\Virpq$ with the central charge $c_{p,q}=13-6p/q-6q/p$\; by the
triplet of the Virasoro primary fields with conformal 
dimension~$\Delta_{1,3}$.

The most investigated models are those with $q=1$.  In this case, the
conformal field theories $\WLM(1,p)$ corresponding to the lattice
models were described in terms of symplectic fermions
in~\cite{[Kausch]} (for $p=2$) and were studied in numerous
papers~\cite{[GK2],M.F,[GaberdielKausch3],[FHST],M.F1,[G-alg],
[Flohr],[GR1],[GR2],[GT],CF,FGK}.
For this set of models, the representation categories of the triplet
algebra $\algW_p$ and of the finite-dimensional quantum group
$\UresSL{2}$~\cite{[FGST]} at the $p$\,th root of unity are equivalent
as braided tensor categories~\cite{[FGST2]}. This is the manifestation
of the Kazhdan--Lusztig duality between vertex-operator algebras and
quantum groups in logarithmic models. This duality means that\, (i)
there is a one-to-one correspondence between representations; (ii) fusion
rules of a conformal model can be calculated by tensor products of a
quantum group representations and (iii) the modular group action 
generated from
chiral characters coincides with the one on the center of the
corresponding quantum group. In the logarithmic models $\WLM(1,p)$,
the Kazhdan--Lusztig duality is presented in its full strength (see
also review~\cite{Semikh-rew}). In
particular, the fusion calculated in~\cite{[FHST]} (see also~\cite{KF}) 
coincides with the Grothendieck ring of the $\UresSL{2}$.
 
For general coprime $p$ and $q$, the models $\WLM(p,q)$ also
demonstrate the Kazhdan--Lusztig duality with a quantum
group~\cite{[FGST3]} but relation between the quantum group and the
$\algW_{p,q}$ algebra is more subtle. There is \textit{no} one-to-one
correspondence between representations but the modular group action on
the center~\cite{[FGST4]} coincides with the one on chiral characters
in the $\algW_{p,q}$ theory and the quantum-group
fusion~\cite{[FGST3]} coincides with the fusion~\cite{JRas-W} of
$\algW_{p,q}$ representations under the identification of the fusion
generators $\mathcal{K}^+_{1,2}{\to}(1,2)_{\mathcal{W}}$ and
$\mathcal{K}^+_{2,1}\to (2,1)_{\mathcal{W}}$; we also identify
$\mathcal{K}^+_{a,b}\to (a,b)_{\mathcal{W}}$,
$\mathcal{K}^-_{a,b}{\to}(\Delta_{a,3q-b})_{\mathcal{W}}$~\footnote{
Indecomposable rank-2 and rank-3 representations appearing in the
fusion~\cite{JRas-W} are identified with the corresponding
quantum-group modules~\cite{[FGST4]} in the following way:
$\mathcal{P}^{+,+}_{a,b}\to
(\mathcal{R}^{p-a,0}_{p,b})_{\mathcal{W}}$,
$\mathcal{P}^{+,-}_{a,b}{\to}(\mathcal{R}^{0,q-b}_{a,q})_{\mathcal{W}}$
and $\mathsf{P}^{+}_{a,b}\to
(\mathcal{R}^{p-a,q-b}_{p,q})_{\mathcal{W}}$, $\mathsf{P}^{-}_{a,b}\to
(\mathcal{R}^{p-a,q-b}_{2p,q})_{\mathcal{W}}$.}. This fusion was also
derived~\cite{[Semi]} from the modular group properties of the chiral
characters.

Other choice of boundary conditions in the lattice models~\cite{[PRZ]}
leads to logarithmic conformal field models $\LM(p,q)$ with the
Virasoro symmetry $\Virpq$. Fusion rules for these models were
calculated in~\cite{RPfus} using a lattice approach and for some cases
in~\cite{EBF,[GaberdielKausch]} using the Nahm algorithm and 
in~\cite{RS-Q} using quantum-group symmetries in XXZ models at a root of unity.

In the present paper, we propose using the Kazhdan--Lusztig duality in
calculating the fusion rules for the simplest subset $\LM(1,p)$ of the
$\LM(p,q)$ models. 
We construct a quantum group dual to the Virasoro algebra $\Vir$ as an
extension of the quantum group $\UresSL2$ dual to $\algW_p$ from the
$\WLM(1,p)$ models. This quantum group is the Lusztig limit
$\LUresSL2$ of the usual quantum $s\ell(2)$
as~$\q\to e^{\imath\pi/p}$ and has the set
of irreducible representations $\XX^{\alpha}_{s,r}$, where
$s=1,2,\dots,p$\, and $\alpha=\pm$\, are $\UresSL2$ highest weight
parameters and $\frac{r-1}{2}$, $r\in\oN$, is the $s\ell(2)$ spin (see
precise definitions in Sec.~\bref{sec:irreps}).  The module
$\XX^{\alpha}_{s,r}$ is a tensor product of $s$-dimensional
irreducible $\UresSL2$- and $r$-dimensional irreducible
$s\ell(2)$-modules.  To each $\XX^{\alpha}_{s,r}$, a projective cover
$\PP^{\alpha}_{s,r}$ corresponds and
$\PP^{\alpha}_{p,r}=\XX^{\alpha}_{p,r}$.  The set of irreducible and
projective modules is closed under tensor products.

We show that the fusion \cite{RPfus} of irreducible and logarithmic
$\Vir$-representations coincides with tensor products of $\LUresSL2$
irreducible and projective modules.  To formulate the main result of
the paper, we introduce the following sum notations
\begin{gather*}
\mathop{\bigoplus{\kern-3pt}'}\limits_{r=a}^{b}f(r) = 
\bigoplus_{r=a}^{b}(1 - \ffrac{1}{2}\delta_{r,a} - \ffrac{1}{2}\delta_{r,b})f(r),\\
\mathop{\bigoplus{\kern-3pt}''}\limits_{r=a}^{b}f(r) =
\bigoplus_{r=a}^{b} \bigl(1 - \ffrac{3}{4}\delta_{r,a} -
\ffrac{1}{4}\delta_{r,a+2}(1+\delta_{a,-1}) - \ffrac{1}{4}\delta_{r,b-2} -
\ffrac{3}{4}\delta_{r,b}\bigr)f(r).
\end{gather*}

\begin{Thm}\label{thm-main}
 The tensor products between irreducible $\LUresSL2$-modules are
\begin{equation*}
\repX^{\alpha}_{s_1,r_1}\tensor\repX^{\beta}_{s_2,r_2} =
\bigoplus_{\substack{r=|r_1-r_2|+1\\step=2}}^{r_1+r_2-1}
\Bigr(\bigoplus_{\substack{s=|s_1-s_2|+1\\step=2}}^{\substack{
\min(s_1 + s_2 - 1,\\ 2p - s_1 - s_2 - 1)}}\!\!\!\repX^{\alpha \beta}_{s,r}
\;+\!\!\! \bigoplus_{\substack{s=2p - s_1 - s_2 +1\\step=2}}^{p-\gamma_2}
\!\!\!\!\!\!\PP^{\alpha \beta}_{s,r}\Bigl)
\end{equation*}
between the irreducible and projective modules are
\begin{equation*}
\repX^{\alpha}_{s_1,r_1}\tensor\PP^{\beta}_{s_2,r_2} = \!\!\!\!\!\!\!
\bigoplus_{\substack{r=|r_1-r_2|+1\\step=2}}^{r_1+r_2-1}\Bigl(
\bigoplus_{\substack{s=|s_1-s_2|+1\\step=2}}^{\substack{
\min(s_1 + s_2 - 1,\\ 2p - s_1 - s_2 - 1)}}\!\!\!\!\!\!\PP^{\alpha \beta}_{s,r}
+2\!\!\!\!\!\!\bigoplus_{\substack{s=2p-s_1-s_2+1\\step=2}}^{p-\gamma_2}
\!\!\!\!\!\!\PP^{\alpha\beta}_{s,r}\Bigr)
+2\mathop{\bigoplus{\kern-3pt}'}\limits_{\substack{r=|r_1-r_2|\\step=2}}^{r_1+r_2}
\bigoplus_{\substack{s=p-s_1+s_2+1\\step=2}}^{p-\gamma_1}
\!\!\!\!\!\!\PP^{-\alpha \beta}_{s,r},
\end{equation*}
and between the projective modules are
\begin{multline*}
\PP^{\alpha}_{s_1,r_1}\tensor \PP^{\beta}_{s_2,r_2}
= 2 \bigoplus_{\substack{r=|r_1-r_2|+1\\\text{step}=2}}^{r_1+r_2-1} 
\Bigl(\bigoplus_{\substack{s=|s_1-s_2|+1\\\text{step}=2}}^{\substack{
\min(s_1 + s_2 - 1,\\ 2p - s_1 - s_2 - 1)}}\PP^{\alpha\beta}_{s,r}
+2\bigoplus_{\substack{s=2p-s_1 - s_2 + 1\\\text{step}=2}}^{p-\gamma_2}
\PP^{\alpha\beta}_{s,r}\Bigr)\\
+2\mathop{\bigoplus{\kern-3pt}'}\limits_{\substack{r=|r_1-r_2|\\
\text{step}=2}}^{r_1+r_2}\Bigl(
\bigoplus_{\substack{s=|p-s_1-s_2|+1\\\text{step}=2}}^{\substack{
\min(p-s_1 + s_2 - 1,\\ p + s_1 - s_2 - 1)}}\!\!\!\!\!\!\PP^{-\alpha\beta}_{s,r}
+2\!\!\!\!\!\!\bigoplus_{\substack{s=\min(p-s_1 + s_2 + 1,\\ p + s_1 - s_2 + 1)}}^
{p-\gamma_1}\!\!\!\!\!\!\PP^{-\alpha\beta}_{s,r}\Bigl)
+\,4\!\!\!\mathop{\bigoplus{\kern-3pt}''}\limits_{\substack{r=|r_1-r_2|-1\\step=2}}^{r_1+r_2+1}
\bigoplus_{\substack{s=s_1 + s_2 + 1\\\text{step}=2}}^{p-\gamma_2}
\!\!\!\!\!\!\PP^{\alpha\beta}_{s,r},
\end{multline*}
where we set $\gamma_1=(s_1+s_2+1)\!\!\!\!\mod2, \;
\gamma_2=(s_1+s_2+p+1)\!\!\!\!\mod2$.
\end{Thm}
The $\LUresSL2$ representation category $\catC_p$ is a direct sum of
two full subcategories $\catC_p=\catC_p^+\oplus\catC_p^-$ such that
there are no morphisms between $\catC_p^+$ and $\catC_p^-$ and the
subcategory $\catC_p^+$ is closed under tensor products. The set of
irreducible modules belonging to the subcategory $\catC_p^+$ is
exhausted by the irreducible modules $\XX^{\alpha}_{s,r}$ with
$\alpha=+$ whenever $r$ is \textit{odd}, and $\alpha=-$ whenever $r$ is
\textit{even}.

The category $\catC_p^+$ is equivalent as a tensor category to the
category of Virasoro algebra representations appearing in
$\LM(1,p)$. We do not describe here this Virasoro category but note
only that under this equivalence irreducible and projective modules
are identified in the following way
\begin{equation}\label{identification}
\begin{split}
\repX^+_{p,2r-1}\to \modR^0_{2r-1},\quad
\repX^-_{p,2r}\to \modR^0_{2r},\quad
\PP^+_{s,2r-1}\to \modR^{p-s}_{2r-1},\quad
\PP^-_{p-s,2r}\to \modR^{s}_{2r},\\
\repX^+_{s,2r-1}\to(2r-1,s),\qquad
\repX^-_{s,2r}\to(2r,s), \qquad 1\leq s\leq p,\quad r\geq 1,
\end{split}
\end{equation}
where $(r,s)$ are the irreducible Virasoro modules with the heighest
weights $\Delta_{r,s} = ((pr-s)^2-(p-1)^2)/4p$ and the $\modR^{s}_{r}$
are logarithmic Virasoro modules from $\LM(1,p)$.  Under this
identification, the fusion \cite{RPfus} for $\LM(1,p)$ is given by
tensor products of the corresponding $\LUresSL2$ representations.

The quantum groups dual to the logarithmic conformal models $\LM(1,p)$
as well as $\WLM(1,p)$ can be constructed in the free field
approach~\cite{[FGST],[FGST3]}. In this approach, the corresponding
quantum groups and the chiral algebras are mutual \textit{maximal
centralizers} of each other. To construct the quantum group for
$\LM(1,p)$, we first note that the chiral algebras $\algW_p$ realized
in the $\WLM(1,p)$ models admit $SL(2)$-action by
symmetries. Invariants of this action is the universal enveloping of
the Virasoro algebra $\Vir$ with
the central charge
  $c_p=13-6p-6/p$.
This suggests that a quantum group $\LUresSL2$ dual to the Virasoro
algebra from $\LM(1,p)$ should be a combination of the quantum group
$\UresSL2$ dual to $\algW_p$ and ordinary $s\ell(2)$.

The quantum group $\LUresSL2$ is constructed in the free-field
approach using the screening operators of the Virasoro algebra $\Vir$
with the central charge $c_p$. We recall there are two screening
operators $e=\oint e^{\sqrt{2p}\,\varphi(z)}dz$\, and $F=\oint
e^{-\sqrt{\frac{2}{p}}\,\varphi(z)}dz$\, commuting with $\Vir$ and the
screening $F$ commutes~\cite{[FHST]} with the extended chiral algebra
$\algW_p$ and generates the lower-triangular part of the $\UresSL2$
with the relation $F^p=0$. Then, considering the deformation
$F_\epsilon=\oint e^{(-\sqrt{\frac{2}{p}}+\epsilon) \varphi(z)}dz$, we
can construct an operator $f=\lim\limits_{\epsilon\to0}
\frac{F_\epsilon^p}{\epsilon}$.  The operators $e$ and $f$ generate
the $s\ell(2)$ from the previous paragraph. To obtain a Hopf-algebra
structure on $\LUresSL2$, we use here the purely algebraic approach
following Lusztig.  We construct the quantum group $\LUresSL2$ as a
limit of the quantum group $U_{\q}(s\ell(2))$ as $\q\to
e^{\frac{\imath\pi}{p}}$.  There is an evident limit in which $E^p$,
$F^p$ and $K^{p}$ become central but we consider another limit in
which
the relations $E^p=F^p=0$, $K^{2p}=1$ are imposed but the generators
$e=\frac{E^p}{[p]!}$ and $f=\frac{F^p}{[p]!}$ are kept in the limit.
In the limit $\q\to e^{\frac{\imath\pi}{p}}$, we have $[p]!=0$ and the
ambiguity~$\frac{0}{0}$ is solved in such a way that the $e$ and $f$
become generators of the ordinary $s\ell(2)$. We thus obtain a Hopf
algebra $\LUresSL2$ that contains the quantum group $\UresSL2$ as a
Hopf ideal and the quotient is the $U(s\ell(2))$, the universal
enveloping of the $s\ell(2)$.

It is noteworthy to mention that a similar duality is also presented
in quantum spin chains (XXZ models) with nondiagonalizable action of
the Hamiltonian. Here, there are two commuting actions of a
Temperley--Lieb algebra (or a proper its extension) and
of a quantum group. In other words, the space of spin states is a
bimodule over the Temperley--Lieb algebra (or its extension) and the
corresponding quantum group~\cite{RS-Q}. Moreover, the quantum
group symmetries are stable with respect to increasing the number of
sites and should be kept in a scaling limit; fusion rules for
Temperley--Lieb algebra representations are obtained by an induction
procedure joining two chains end to end and can be calculated using
only the quantum group symmetries~\cite{RS-Q,RS2}. Thus, these
two algebraic objects
are in some duality which is a lattice version of the Kazhdan--Lusztig
duality.

We also note that tensor products of $\LUresSL2$ projective modules
reproduce the Virasoro fusion rules proposed in~\cite{RS-Q} under the
identifications:
\begin{align*}
&p=2:&\qquad
\PP^+_{1,2r-1} \to \modVir_{2r-1}, \qquad\PP^-_{1,2r} \to
    \modVir_{2r},\qquad r\in\oN,\\
&p=3:&\quad\XX^+_{3,2r-1} \to \modVir_{3r-2},\qquad
\PP^+_{1,2r-1} \to \modVir_{3r-1},\qquad
\PP^-_{2,2r} \to \modVir_{3r},\qquad r\in\oN.
\end{align*}

The paper is organized as follows. In Sec.~\ref{sec:QG-VOA}, we
introduce the quantum group $\LUresSL2$ dual to the Virasoro algebra
$\Vir$ and describe a Hopf algebra structure on $\LUresSL2$.  In
Sec.~\ref{sec:LU-rep}, we describe irreducible representations of
$\LUresSL2$, calculate all possible extensions between them and then
construct projective modules. In Sec.~\ref{sec:fusion}, we
decompose tensor products between irreducible and projective
$\LUresSL2$-modules.

\section{Quantum groups as centralizers of VOAs.\label{sec:QG-VOA}}
In this section, we introduce a quantum group that commutes with the
Virasoro algebra $\Vir$ on the chiral space of states in the free
massless scalar field theory
\begin{equation*}
\varphi(z)\varphi(w) = \log(z-w)
\end{equation*}
with the energy-momentum tensor
\begin{equation}
 T=\half\partial\varphi\partial\varphi+\frac{\alpha_0}{2}\partial^2\varphi,
\end{equation}
where the background charge $\alpha_0 = \alpha_+ + \alpha_- =
\sqrt{2p} - \sqrt{2/p}$. This quantum group is denoted as $\LUresSL2$
and constructed as some extension of the finite-dimensional quantum
group $\UresSL2$ which is the maximal centralizer of the triplet
algebra $\Walg$.  We recall that $\Walg$ is an extension of $\Vir$ by
the $s\ell(2)$-triplet of the fields~$W^{\pm,0}(z)$~\cite{[FHST]}:
\begin{equation*}
  W^-(z):=\rme^{-\alpha_+\varphi}(z),\quad\;
  W^0(z):=[S_+,W^-(z)],\quad\;
  W^+(z):=[S_+,W^0(z)],
\end{equation*}
where $S_+$ is the ``long screening'' operator
$\oint\rme^{\alpha_+\varphi}dz$. These three fields are Virasoro
primaries and their conformal dimensions equal to $(2p{-}1)$. The
field $W^-(z)$ is the lowest-weight vector (with respect to the Cartan
$s\ell(2)$-generator $h=\frac{1}{\alpha_+}\, \varphi_0$ and $\varphi_0$
is the zero-mode of $\partial\varphi(z)$), the field $W^0(z)$ has the
$h$-weight equals to~$0$ and $W^+(z)$ is the highest-weight vector of
the $s\ell(2)$-triplet. 

Irreducible representations of the triplet algebra $\Walg$ admit two
commuting actions, $s\ell(2)$- and $\Vir$-actions~\cite{[FGST3]}, and
the Virasoro algebra $\Vir$ is the invariant of the $s\ell(2)$ action
in the vacuum representation of $\Walg$. This suggests a construction
of the maximal centralizer for $\Vir$ as an extension of the
centralizer $\UresSL2$ for the triplet algebra $\Walg$ by the
$s\ell(2)$ triplet: $e=S_+$, $h=\frac{1}{\alpha_+}\, \varphi_0$, and a
``conjugate'' operator $f$ to the long screening $S_+$.

\subsection{The centralizer of $\Walg$} We recall the definition
of the quantum group $\UresSL2$~\cite{[FGST]} that commutes with the
triplet algebra $\Walg$ action on the chiral space of states.  The
$\UresSL2$ can be constructed as the Drinfeld double of the Hopf
algebra generated by the ``short screening'' operator
$F=\oint\rme^{\alpha_-\varphi(z)}dz$ and
$K=\rme^{-i\pi\alpha_-\varphi_0}$. The Hopf algebra structure is found
from the action of these operators on fields.  In particular, the
comultiplication is calculated from the action of $F$ and $K$ on
operator product expansions of fields. Details of constructing
$\UresSL2$ are given in~\cite{[FGST]}.

The quantum group $\UresSL2$ is the ``restricted'' quantum $s\ell(2)$
with $\q = e^{i\pi/p}$ and the three generators $E$, $F$, and
$K$ satisfying the standard relations for the quantum $s\ell(2)$,
\begin{equation}\label{Uq-com-relations}
  KEK^{-1}=\q^2E,\quad
  KFK^{-1}=\q^{-2}F,\quad
  [E,F]=\ffrac{K-K^{-1}}{\q-\q^{-1}},
\end{equation}
with some additional constraints,
\begin{equation}\label{root-1-rel}
  E^{p}=F^{p}=0,\quad K^{2p}=\one,
\end{equation}
and the Hopf-algebra structure is given by
\begin{gather}
  \Delta(E)=\one\otimes E+E\otimes K,\quad
  \Delta(F)=K^{-1}\otimes F+F\otimes\one,\quad
  \Delta(K)=K\otimes K,\label{Uq-comult-relations}\\
  S(E)=-EK^{-1},\quad  S(F)=-KF,\quad S(K)=K^{-1},
  \label{Uq-antipode}\\
  \epsilon(E)=\epsilon(F)=0,\quad\epsilon(K)=1.\label{Uq-epsilon}
\end{gather}

The quantum group $\UresSL2$ admits $(p-1)$ dimensional family of
inequivalent extensions by $s\ell(2)$ algebras acting on $\UresSL2$ as
exterior derivatives~\cite{[DFGT]}. There is an extension that admits a
Hopf algebra structure to be defined in the following subsection.

\subsection{The centralizer of $\Vir$} Here, we define a quantum group 
$\LUresSL2$ (i.e. a Hopf algebra) that commutes with the Virasoro
algebra $\Vir$ on the chiral space of states.

\subsubsection{Definition} The Hopf-algebra
structure on $\LUresSL2$ is the following. The defining relations
between the $E$, $F$, and $K$ generators are the same as in $\UresSL2$
and given in~\eqref{Uq-com-relations} and~\eqref{root-1-rel},
and the usual $s\ell(2)$ relations between the $e$, $f$, and $h$:
\begin{equation}\label{sl2-rel}
  [h,e]=e,\qquad[h,f]=-f,\qquad[e,f]=2h,
\end{equation}
and the ``mixed'' relations
\begin{gather}
  [h,K]=0,\qquad[E,e]=0,\qquad[K,e]=0,\qquad[F,f]=0,\qquad[K,f]=0,\label{zero-rel}\\
  [F,e]= \ffrac{1}{[p-1]!}K^p\ffrac{\q K-\q^{-1} K^{-1}}{\q-\q^{-1}}E^{p-1},\label{Fe-rel}\\
  [E,f]=\ffrac{(-1)^{p+1}}{[p-1]!} F^{p-1}\ffrac{\q K-\q^{-1} K^{-1}}{\q-\q^{-1}},
    \label{Ef-rel}\\
  [h,E]=\frac{1}{2}EA,\quad[h,F]=- \frac{1}{2}AF,\label{hE-hF-rel}
\end{gather}
where 
\begin{equation}\label{A-element}
  A=\,\sum_{s=1}^{p-1}\ffrac{(u_s(\q^{-s-1})-u_s(\q^{s-1}))K
        +\q^{s-1}u_s(\q^{s-1})-\q^{-s-1}u_s(\q^{-s-1})}{(\q^{s-1}
         -\q^{-s-1})u_s(\q^{-s-1})u_s(\q^{s-1})}\,
        u_s(K)\idem_s
\end{equation}
with $u_s(K)=\prod_{n=1,\;n\neq s}^{p-1}(K-\q^{s-1-2n})$, and
$\idem_s$ are the central primitive idempotents of $\UresSL2$ given in
App.~\bref{app:idem}.

The comultiplication in $\LUresSL2$ is given
in~\eqref{Uq-comult-relations} for the $E$, $F$, and $K$ generators
and
\begin{gather}
  \Delta(e)=e\tensor1+K^p\tensor e
  +\ffrac{1}{[p-1]!} \sum_{r=1}^{p-1}\frac{\q^{r(p-r)}}{[r]}K^pE^{p-r}\tensor E^r
  K^{-r},\label{e-comult}\\
 \Delta(f)= f\tensor 1+K^p\tensor f+\ffrac{(-1)^p}{[p-1]!} 
  \sum_{s=1}^{p-1}\frac{\q^{-s(p-s)}}{[s]}K^{p+s}F^s\tensor F^{p-s},\label{f-comult}
\end{gather}
an explicit form of $\Delta(h)=\half[\Delta(e),\Delta(f)]$ is very
bulky and we do not give it here. 

The antipode $S$ and the counity $\epsilon$ are given
in~\eqref{Uq-antipode}-\eqref{Uq-epsilon} and
\begin{gather}
  S(e)=-K^pe,\qquad S(f)=-K^pf,\qquad S(h)=-h,\\
  \epsilon(e)=\epsilon(f)=\epsilon(h)=0.\label{relations-end}
\end{gather}


\subsubsection{$\LUresSL2$ through divided powers} The quantum group
$\LUresSL2$ can be realized as the Lusztig extension of the restricted
quantum group $\UresSL2$ by divided powers of the $E$ and $F$
generators. In the usual quantum $\UqSL2$ with the relations
\eqref{Uq-com-relations}, \eqref{Uq-comult-relations},
\eqref{Uq-antipode}, and \eqref{Uq-epsilon} and with generic $\q$, the
$e$, $f$ and $h$ exist as the following elements,
\begin{equation*}
  e=\ffrac{1}{[p]!}K^pE^p,\qquad
  f=\ffrac{(-1)^p}{[p]!}F^p,
\end{equation*}
and the Cartan element
\begin{multline*}
 \!\!\! h=\half[e,f]=
  \ffrac{1-\q^{2p}}{2}ef+\ffrac{\q^{2p}}{2p(\q-\q^{-1})}
  \Bigl(\ffrac{\one - (-\q^p)^{p-1}K^{2p}}{[p]}+
  \sum_{r=1}^{p-1}\ffrac{(-1)^r\q^{r(p-1)}[p-1]!}{[p-r]![r]!}K^{2r}\Bigr)+\\
  +\ffrac{(-1)^{p+1}}{2[p-1]!}
  K^p\sum_{n=1}^{p-1}\ffrac{(-1)^n[p-1]!}{([p-n]!)^2[n]!}
  \prod_{k=0}^{p-n-1}\ffrac{\cas-\q^{-(2k+1)}K-\q^{(2k+1)}K^{-1}}{(\q-\q^{-1})^2}
  \prod_{r=1}^n\ffrac{\q^{r-1}K-\q^{-r+1}K^{-1}}{\q-\q^{-1}},
\end{multline*}
where the $\UqSL2$ Casimir element $\cas$ is given in App.~\bref{app:idem}.
The commutation relations between these elements and their
comultiplication, antipode, and counity satisfy
\eqref{sl2-rel}-\eqref{relations-end} when $\q\to e^{i\pi/p}$. 

\section{Representations of $\LUresSL2$\label{sec:LU-rep}}
To describe the category $\catC_p$ of finite-dimensional
$\LUresSL2$-modules, we first study irreducible $\LUresSL2$-modules
in~\bref{sec:irreps} and then obtain essential information about
possible extensions between them in~\bref{sec:exts}. This let us
construct all finite-dimensional projective $\LUresSL2$-modules
in~\bref{proj-mod}. Then in~\bref{sec:cat-decomp}, we decompose the
representation category $\catC_p$ into a direct sum of two full
subcategories $\catC_p^+$ and $\catC_p^-$. The subcategory $\catC_p^+$
is then identified with a tensor category of the Virasoro algebra
representations.

\subsection{Irreducible $\LUresSL2$-modules}\label{sec:irreps}
An irreducible $\LUresSL2$-module $\repX^{\pm}_{s,r}$ is labeled by
$(\pm,s,r)$, with $1\leq s\leq p$ and $r\in\oN$, and has the highest
weights $\pm\q^{s-1}$ and $\frac{r-1}{2}$ with respect to $K$ and $h$
generators, respectively. The $sr$-dimensional module $\XX^{\pm}_{s,r}$
is spanned by elements $\stprp_{n,m}^{\pm}$, $0\leq n\leq s{-}1$,
$0\leq m\leq r{-}1$, where $\stprp_{0,0}^{\pm}$ is the highest-weight
vector and the left action of the algebra on $\XX^{\pm}_{s,r}$ is
given~by
\begin{align}
  K \stprp_{n,m}^{\pm} &=
  \pm \q^{s - 1 - 2n} \stprp_{n,m}^{\pm},\qquad
  &h\, \stprp_{n,m}^{\pm} &=  \half(r-1-2m)\stprp_{n,m}^{\pm},\label{basis-lusz-irrep-1}\\
  E \stprp_{n,m}^{\pm} &=
  \pm [n][s - n]\stprp_{n - 1,m}^{\pm},\qquad
  &e\, \stprp_{n,m}^{\pm} &=  m(r-m)\stprp_{n,m-1}^{\pm},\label{basis-lusz-irrep-2}\\
  F \stprp_{n,m}^{\pm} &= \stprp_{n + 1,m}^{\pm},\qquad
  &f\, \stprp_{n,m}^{\pm} &=  \stprp_{n,m+1}^{\pm},\label{basis-lusz-irrep-3}
\end{align}
where we set $\stprp_{-1,m}^{\pm}=\stprp_{n,-1}^{\pm}
=\stprp_{s,m}^{\pm} =\stprp_{n,r}^{\pm}=0$.

\begin{rem}
The element $A$ defined in~\eqref{A-element}
and~\eqref{hE-hF-rel} is represented in an irreducible representation
of $\LUresSL2$ by an operator acting as the identity on the
highest-weight vector and zero on all other vectors. Therefore, as it
follows from the relations~\eqref{zero-rel}-\eqref{hE-hF-rel}, the
$E$, $F$, and $K$ generators of the subalgebra $\UresSL2$ commute on
$\XX^{\pm}_{s,r}$ with the $e$, $f$, and $h$ generators of the
subalgebra $s\ell(2)$.
\end{rem}

\subsection{Extensions among irreducibles}\label{sec:exts}
Here, we study possible extensions between irreducible
$\LUresSL2$-modules to construct indecomposable modules in what
follows. Let $A$ and $C$ be left $\LUresSL2$-modules.  We say that a
short exact sequence of $\LUresSL2$-modules $0\to A\to B\to C\to 0$
is an \textit{extension} of $C$ by $A$, and we let $\Ext(C,A)$ denote
the set of equivalence classes (see, e.g.,~\cite{[M]}) of extensions
of $C$ by~$A$.

\begin{lemma}\label{lemma:exts}
For $1 \leq s\leq p-1$, $r\in \oN$ and $\alpha,\alpha'\,{\in}\,\{+,-\}$,
  there are vector-space isomorphisms
  \begin{equation*}
    \Ext(\repX^{\alpha}_{s,r},\repX^{\alpha'}_{s',r'})\cong
    \begin{cases}
      \oC,\quad \alpha'=-\alpha, \; s'=p-s, \; r'=r\pm 1,\\
      0,\quad \text{otherwise}.
    \end{cases}
  \end{equation*}
There are no nontrivial extensions between $\XX^{\pm}_{p,r}$ and any
irreducible module.
\end{lemma}

\begin{proof}
We first recall~\cite{[FGST2]} that the space $\ExtU$ of extensions
between irreducible modules over the subalgebra $\UresSL2$ is at most
two-dimensional and there exists a nontrivial extension only between
$\repX^{\pm}_s$ and $\repX^{\mp}_{p-s}$, where $1 \leq s\leq p-1$ and
we set $\repX^{\pm}_s{=}\repX^{\pm}_{s,1}|_{\UresSL2}$. Moreover,
there is an action of $\LUresSL2$ on projective resolutions for
irreducible $\UresSL2$-modules and this generates an action of the
quotient-algebra $s\ell(2)$ on the corresponding cochain complexes and
their cohomologies. Therefore, for an irreducible $\repX$ and an
$\UresSL2$-module $\modM$, all extension groups
$\ExtUbul(\repX,\modM)$ are $s\ell(2)$-modules. In particular, the
space $\ExtU(\repX^{\pm}_s,\repX^{\mp}_{p-s})$ is the
$s\ell(2)$-\,doublet and there is a nontrivial $s\ell(2)$ action on
all Hochschild cohomologies of $\UresSL2$ (see~\cite{[DFGT]}).

Next, to calculate the first extension groups between the irreducible
$\LUresSL2$-modules, we use the Serre-Hochschild spectral sequence
with respect to the subalgebra $\UresSL2$ and the quotient-algebra
$s\ell(2)$. The spectral sequence is degenerate at the second term due
to the semisimplicity of the quotient algebra and we thus obtain
\begin{equation*}
\Ext(\repX^{\alpha}_{s,r},\repX^{\alpha'}_{s',r'}) = H^0(s\ell(2),
\ExtU(\repX^{\alpha}_{s,r},\repX^{\alpha'}_{s',r'})),
\end{equation*}
where the right-hand side is the vector space of the
$s\ell(2)$-invariants in the $s\ell(2)$-module
$\ExtU(\repX^{\alpha}_{s,r},\repX^{\alpha'}_{s',r'})$. This module is
nonzero only in the case $\alpha'=-\alpha$, $s'=p-s$ and isomorphic to
the tensor product $\repX_2\tensor\repX_r\tensor\repX_{r'}$ of the
$s\ell(2)$ modules, where $\repX_r$ is the $r$-dimensional
module. Obviously, the tensor product contains a trivial
$s\ell(2)$-module only in the case $r'=r\pm 1$. This completes the proof.
\end{proof}

\subsubsection{Example} 
As an example, we describe an extension of $\XX^{\pm}_{s,r}$ by
$\XX^{\mp}_{p-s,r+1}$. This can be realized as an extension of $r$
copies $\repX^{\pm}_s$ by $(r+1)$ copies $\repX^{\mp}_{p-s}$ of the
irreducible modules over the subalgebra $\UresSL2$,
 \begin{equation*}\label{schem-Weyl}
   \xymatrix@=12pt{
     &\stackrel{\repX^{\pm}_{s}}{\circ}\ar@/_/[dr]^{F} \ar@/^/[dl]_{E} \ar[rr]^{f}&
     &\stackrel{\repX^{\pm}_{s}}{\circ}\ar@/^/[dl]_{E}
     \ar@/_/[dr]^{F} \ar[rr]^{f}&&\dots\ar@/^/[dl]_{E}\ar@/_/[dr]^{F} \ar[rr]^{f}&
     &\stackrel{\repX^{\pm}_{s}}{\circ}
     \ar@/^/[dl]_{E} \ar@/_/[dr]^{F}&\\
     \stackrel{\repX^{\mp}_{p-s}}{\bullet} \ar[rr]_{f}&&
     \stackrel{\;\repX^{\mp}_{p-s}\;}{\bullet} \ar[rr]_{f}&&
     \stackrel{\;\repX^{\mp}_{p-s}\;}{\bullet}\ar[r]_{f}&\dots\ar[r]_{f}&
     \stackrel{\;\repX^{\mp}_{p-s}\;}{\bullet} \ar[rr]_{f}&&
     \stackrel{\repX^{\mp}_{p-s}}{\bullet}
   }
 \end{equation*}
with the indicated action of the $E$ and $F$ generators mixing
$\repX^{\pm}_s$ with $\repX^{\mp}_{p-s}$ modules, and with the $e$ and
$f$ generators mapping different copies with the same sign. The
extension thus constructed can be depicted as
 \begin{equation*}\label{schem-LWeyl}
\xymatrix@=10pt{
     \stackrel{\repX^{\pm}_{s,r}}{\circ}\ar@/^/[dr]&\\
     &\stackrel{\repX^{\mp}_{p-s,r+1}}{\bullet}}
\end{equation*}
with the convention that the arrow is directed to the submodule in the
bottom marked by~$\bullet$, in contrast to the subquotient~$\circ$ in
the top.


\subsection{Projective $\LUresSL2$-modules}\label{proj-mod}
We next construct projective $\LUresSL2$-modules as projective covers
of irreducible modules. A projective cover of an irreducible module is
a ``maximal'' indecomposable module that can be mapped onto the
irreducible. Lem.~\bref{lemma:exts} state that we can ``glue'' two
irreducible modules into an indecomposable module only in the case if
the irreducibles have opposite signs of the $\alpha$-index, the
difference between the two $r$-indexes equals to one and the sum of
the two $s$-indexes is equal to $p$.  Therefore, to construct a
projective cover for $\repX^{\pm}_{s,r}$, $1\leq s\leq p-1$ and $r\geq
2$, we first have to obtain a nontrivial extension
of~$\repX^{\pm}_{s,r}$ by the maximal number of irreducible modules,
that is, by
\begin{equation*}
\repX^{\mp}_{p-s,r-1}\boxtimes
\Ext(\repX^{\pm}_{s,r},\repX^{\mp}_{p-s,r-1}) \oplus
\repX^{\mp}_{p-s,r+1}\boxtimes
\Ext(\repX^{\pm}_{s,r},\repX^{\mp}_{p-s,r+1}),
\end{equation*}
which is
  \begin{equation*}
    0\to \repX^{\mp}_{p-s,r-1}\oplus\repX^{\mp}_{p-s,r+1}\to
    \modM^{\pm}_{s,r} \to \repX^{\pm}_{s,r} \to 0.
  \end{equation*}
where $\modM^{\pm}_{s,r}$ is an indecomposable module. Next, to find
the projective cover of $\repX^{\pm}_{s,r}$, we extend the submodule
$\repX^{\mp}_{p-s,r-1}\oplus\repX^{\mp}_{p-s,r+1}
\subset\modM^{\pm}_{s,r}$ by the maximal number of irreducible
modules, that is, by $\repX^{\pm}_{s,r-2}\oplus
2\repX^{\pm}_{s,r}\oplus \repX^{\pm}_{s,r+2}$.  The compatibility with
the $\LUresSL2$-algebra relations (with $F^p\,{=}\,E^{p}\,{=}\,0$ and
\eqref{hE-hF-rel} in particular) leads to an extension corresponding
to the module $\PP^{\pm}_{s,r}$ with the following subquotient
structure:
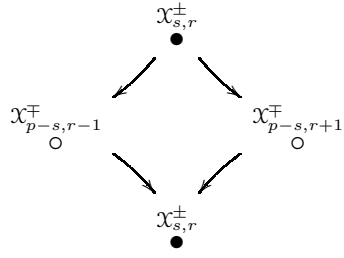
\begin{equation}\label{schem-proj}
  \xymatrix@=12pt{
    &&\stackrel{\XX^{\pm}_{s,r}}{\bullet}
    \ar@/^/[dl]
    \ar@/_/[dr]
    &\\
    &\stackrel{\XX^{\mp}_{p - s, r - 1}}{\circ}\ar@/^/[dr]
    &
    &\stackrel{\XX^{\mp}_{p - s, r + 1}}{\circ}\ar@/_/[dl]
    \\
    &&\stackrel{\XX^{\pm}_{s,r}}{\bullet}&
  }
\end{equation}
and the $\LUresSL2$ action is explicitly described in
App.~\bref{app:proj-mod-base}.  This module being restricted to the
subalgebra $\UresSL2$ is a direct sum of projective modules that
covers the direct sum $\oplus_{i=1}^r \repX^{\pm}_s$, where we set
$\repX^{\pm}_s = \repX^{\pm}_{s,1}|_{\UresSL2}$. Therefore, the
$\PP^{\pm}_{s,r}$ module is the projective cover of
$\repX^{\pm}_{s,r}$, for $1\leq s\leq p-1$ and $r\geq 2$.
Similar procedure gives the projective cover $\PP^{\pm}_{s,1}$ for the
irreducible module $\XX^{\pm}_{s,1}$ with the following subquotient
structure:
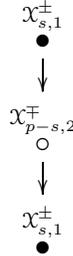
\begin{equation}\label{schem-proj-sm}
  \xymatrix@=12pt{
    &\stackrel{\XX^{\pm}_{s,1}}{\bullet}\ar[d]_{}
    \\
    &\stackrel{\XX^{\mp}_{p - s, 2}}{\circ}\ar[d]^{}
    \\
    &\stackrel{\XX^{\pm}_{s,1}}{\bullet}
  }
\end{equation}
and the $\LUresSL2$ action is also explicitly described in
App.~\bref{app:proj-mod-base}.

A ``half'' of these projective modules is then identified in the
fusion algebra calculated below in Sec.~\bref{sec:fusion} with some
logarithmic Virasoro representations.

\begin{rem}\label{rem:proj-param} We note there are no additional
parameters distinguishing nonisomorphic indecomposable
$\LUresSL2$-modules with the same subquotient structure as
in~\eqref{schem-proj} and~\eqref{schem-proj-sm}. This trivially
follows from the fact that all infinitesimal deformations of
homomorphisms $f:\LUresSL2\to\End(\PP^{\pm}_{s,r})$ continuing
infinitesimal deformations of the algebra $\End(\PP^{\pm}_{s,r})$ are
in one-to-one correspondence with elements in the cohomology space
$H^1(\LUresSL2,\End(\PP^{\pm}_{s,r}))$.
Since the module $\End(\PP^{\pm}_{s,r})$ is a direct sum of projective
modules then the cohomologies space $H^1$ is trivial.
\end{rem}

\subsubsection{Semisimple length of a module} Let $\modN$ be a
$\LUresSL2$-module.  We define a \textit{semisimple filtration} of
$\modN$ as a tower of submodules
\begin{equation*}
  \modN=\modN_0\supset\modN_1\supset\ldots\supset\modN_l=0
\end{equation*}
such that each quotient $\modN_i/\modN_{i+1}$ is semisimple.  The
number $l$ is called the \textit{length} of the filtration.  In the
set of semisimple filtrations of $\modN$, there exists a filtration
with the minimum length~$\ell$.  We call~$\ell$ the \textit{semisimple
  length} of $\modN$.

Evidently, an indecomposable module has the semisimple length not less
than~$2$.  Any semisimple module has the semisimple length~$1$.

\begin{prop}\mbox{}

\begin{enumerate}
\item
Every indecomposable $\LUresSL2$-module with the semisimple length~$3$
is isomorphic to $\PP^{\pm}_{s,r}$\,, for some $s\in \{1, 2, \dots,
p-1\}$ and some finite $r\in\oN$.
\item
There are no indecomposable modules with the semisimple length greater
than $3$.
\end{enumerate}
\end{prop}
\begin{proof}
Consider a module $\modM$ with a nonvanishing mapping
$\PP^{\pm}_{s,r}\to\modM$ that covers $\repX^{\pm}_{s,r}$. In the case
when the mapping is an embedding, we note that the projective module
$\PP^{\pm}_{s,r}$ is also an injective module (the contragredient one
to a projective module) and is therefore a direct summand in any
module into which it is embedded. In the case with a nonvanishing
kernel of the mapping, the kernel contains the submodule
$\repX^{\pm}_{s,r}$ of $\PP^{\pm}_{s,r}$. Therefore, the sub\-quotient
actually belongs to a direct summand in $\modM$ with the semisimple
length $2$~or~$1$.

We thus conclude that there are no indecomposable modules with the
semisimple length $4$ (``higher'' than $\PP^{\pm}_{s,r}$) and every
$\LUresSL2$-module with the semisimple length~$3$ is isomorphic to a
direct sum of $\PP^{\pm}_{s,r}$.
\end{proof}

\subsection{Decomposition of the category $\catC_p$}\label{sec:cat-decomp}
Here, we describe the category $\catC_p$ of finite-dimensional
$\oZ^2$-graded modules over $\LUresSL2$. We use the results about
possible extensions between irreducible modules over $\LUresSL2$,
Lem.~\bref{lemma:exts}, to state the following decomposition theorem.
\begin{thm}\label{thm:cat-decomp}\mbox{}
  \begin{enumerate}
  \item The category $\catC_p$ of finite-dimensional
  $\LUresSL2$-modules has the decomposition
    \begin{equation*}
      \catC_p=\bigoplus_{s=1}^{p-1}\catCpl(s)\oplus\catCmin(s)
      \oplus \bigoplus_{r\in\oN}\catSpl(r)\oplus\catSmin(r),
    \end{equation*}
    where each direct summand is a full indecomposable subcategory.
    
  \item Each of the full subcategories $\catSpl(r)$ and~$\catSmin(r)$ is
    semisimple and contains precisely one irreducible module,
    $\repX^{+}_{p,r}$ and~$\repX^{-}_{p,r}$ respectively.
    
  \item Each $\catCpl(s)$ contains the family of irreducible
    modules~$\repX^{+}_{s,2r-1}$
    and~$\repX^{-}_{p-s,2r}$\,,~$r{\in}\oN$.

  \item Each $\catCmin(s)$ contains the family of irreducible
    modules~$\repX^{+}_{s,2r}$
    and~$\repX^{-}_{p-s,2r-1}$\,,~$r{\in}\oN$.
  \end{enumerate}
\end{thm}

We recall~\cite{[FGST2]} the decomposition of the representation
category $\catCbar_p$ for $\UresSL2$:
\begin{equation}\label{catbar-decomp}
\catCbar_p=\bigoplus_{s=0}^{p}\catCbar(s),
\end{equation}
where each $\catCbar(s)$ is a full indecomposable subcategory and
contains two irreducibles $\repX^+_s$ and $\repX^-_{p-s}$, for $1\leq
s\leq p-1$, that can be composed into an indecomposable module. The
subcategories $\catCbar(0)$ and $\catCbar(p)$ are semisimple and
contain $\repX^+_p$ and $\repX^-_p$ respectively.

\begin{lemma}\label{rem:indecomp-mod-rest}
An indecomposable $\LUresSL2$-module considered a
$\UresSL2$-module is a direct sum of indecomposable modules from the
full subcategory $\catCbar(s)$ in~\eqref{catbar-decomp} for some fixed
$s\in \{1, 2, \dots, p-1\}$.
\end{lemma}
In the following section, we show that the full subcategory
\begin{equation*}
\catCpl_p=\bigoplus_{s=1}^{p-1}\catCpl(s)
      \oplus \bigoplus_{\text{odd}\, r\geq 1}\catSpl(r)
 \oplus \bigoplus_{\text{even}\, r\geq 2}\catSmin(r)
\end{equation*}
in the category $\catC_p=\catCpl_p\oplus\catCmin_p$ is closed under
the tensor product operation and is identified with a tensor category
of representations for the Virasoro algebra $\Vir$. This gives fusion
rules for logarithmic $(1,p)$ models with the chiral symmetry $\Vir$.

\section{The fusion algebra\label{sec:fusion}}
In this section, we calculate tensor products between $\LUresSL2$ irreducible and
projective modules introduced in
Sec.~\bref{sec:LU-rep}. To decompose the tensor products, we first
show how to extend results for $\LUresSL2$-modules with the
$s\ell(2)$-index $r=1$ to modules with~$r>1$.

\begin{Lemma}\label{fus-lemma}
For $1\leq s\leq p$\, and $r\in\oN$, we have isomorphisms of
$\LUresSL2$-modules,
\begin{align*}
&\repX^+_{s,r} \cong \repX^{\pm}_{1,r}\tensor\repX^{\pm}_{s,1}
\cong\repX^{\pm}_{s,1}\tensor\repX^{\pm}_{1,r},
&\repX^-_{s,r} \cong
\repX^{\pm}_{1,r}\tensor\repX^{\mp}_{s,1}
\cong\repX^{\pm}_{s,1}\tensor\repX^{\mp}_{1,r}
\end{align*}
and
\begin{align*}
  &\PP^+_{s,r} \cong \XX^{\pm}_{1,r}\tensor\PP^{\pm}_{s,1} \cong
  \PP^{\pm}_{s,1}\tensor\XX^{\pm}_{1,r}, 
  &\PP^-_{s,r} \cong
  \XX^{\pm}_{1,r}\tensor\PP^{\mp}_{s,1} \cong \PP^{\pm}_{s,1}\tensor\XX^{\mp}_{1,r}.
\end{align*}
\end{Lemma}
\begin{proof}
As an example, we consider $\repX^{+}_{1,r}\tensor\PP^{+}_{s,1}$.
In this case, the comultiplication takes the following form:
\begin{align*}
  &\Delta(E)= \one\otimes E,\quad
  \Delta(F)= K^{-1}\otimes F,\quad
  \Delta(K)= K\otimes K,\\
  &\Delta(e)= K^p\tensor e+e\tensor 1,\quad
  \Delta(f)= K^p\tensor f+f\tensor 1,\quad
  \Delta(h)= 1\tensor h+h\tensor 1.
\end{align*}

Let $\stprp_m$, with $0\leq m\leq r-1$, denotes the basis in
\eqref{basis-lusz-irrep-1} - \eqref{basis-lusz-irrep-3} for
$\repX^{+}_{1,r}$ and $\toppr_n$, $\botpr_n$, $\rightpr_{k,0}$,
$\rightpr_{k,1}$, with $0\leq n\leq s-1$ and $0\leq k\leq p-s-1$, is
the basis in~App.\bref{app:proj-mod-base} for $\PP^{\pm}_{s,1}$. Then,
the basis for the product is
\begin{align*}
&\toppr_{n,m} = \stprp_m\tensor\toppr_n, &\botpr_{n,m} =
  \stprp_m\tensor\botpr_n,& \mbox{}\\
&\rightpr_{k,i} = \stprp_i\tensor\rightpr_{k,0}+i\stprp_{i-1}\tensor\rightpr_{k,1},
&\leftpr_{k,j} = \stprp_j\tensor\rightpr_{k,0}-(r-j)\stprp_{j-1}\tensor\rightpr_{k,1},&
\end{align*}
where $0 \le i \le r$, $1 \le j \le r-1$ and we assume
$\stprp_r\equiv0$. The action of all the generators coincides with the one
in $\PP^+_{s,r}$ explicitly described in App.~\bref{app:proj-mod-base}.
The other cases have a similar proof.
\end{proof}

Then, we calculate tensor products starting with simplest
cases. Tensor products of modules with arbitrary $r$-indexes are based
on their $r=1$ cases and tensor products of a projective module with
an irreducible or a projective module are based on tensor products
of their irreducible subquotients and submodules.

\subsection{Fusion of irreducible modules}
The case of tensor products of two irreducibles with the
$s\ell(2)$-index $r=1$ is the simplest one.
\begin{lemma}
For $1\le s_1,s_2\le p-1$, we have
\begin{equation*}
\repX^{\alpha}_{s_1,1}\tensor\repX^{\beta}_{s_2,1} = 
\bigoplus_{\substack{s=|s_1-s_2|+1\\\step=2}}^{\min(s_1+s_2-1,2p - s_1 -
  s_2 - 1)}
\repX^{\alpha \beta}_{s,1}+\bigoplus_{\substack{s=2p - s_1 - s_2
    +1\\\step=2}}^{p-\gamma_2}
\mathscr{P}^{\alpha \beta}_{s,1},\\
\end{equation*}
where $\gamma_2=(s_1+s_2+p+1)\mod2$.
\end{lemma}
\begin{proof}
The tensor product restricted to the subalgebra $\UresSL2$ has the
well known decomposition~\cite{[FGST],Semikh-rew} (see also~\cite{Erd}
for more general case of Taft Hopf algebras). In this decomposition,
each two direct summands belongs to pairwise different indecomposable
subcategories in the category $\catCbar_p$ (see~\eqref{catbar-decomp})
and according to Lem.~\bref{rem:indecomp-mod-rest} any direct sum of
them can not be combined into an indecomposable
$\LUresSL2$-module. Therefore, the action of the $e$, $f$, and $h$
generators is unambiguously defined.
\end{proof}
Using Lem.~\bref{fus-lemma}, we can easily extend this result to
arbitrary $r$-index.
\begin{thm}\label{fus-irr}
For $1\le s_1\le p-1$, $r_1,r_2\in\oN$, we have
\begin{equation*}
  \repX^{\alpha}_{s_1,r_1}\tensor\repX^{\beta}_{s_2,r_2} =
  \bigoplus_{\substack{r=|r_1-r_2|+1\\\step=2}}^{r_1+r_2-1}
  \Bigr(\bigoplus_{\substack{s=|s_1-s_2|+1\\\step=2}}^{\min(s_1+s_2-1,2p -
  s_1 - s_2 - 1)}\repX^{\alpha \beta}_{s,r}+\bigoplus_{\substack{s=2p - s_1
  - s_2 +1\\\step=2}}^{p-\gamma_2}\mathscr{P}^{\alpha
  \beta}_{s,r}\Bigl),\\
\end{equation*}
where $\gamma_2=(s_1+s_2+p+1)\mod2$.
\end{thm}
\begin{proof}
We only need to show that\, $\XX^+_{1,r_1}\tensor\XX^+_{1,r_2}
=\bigoplus_{\substack{r=|r_1-r_2|+1\\step=2}}^{r_1+r_2-1}
\XX^+_{1,r}.$ This trivially follows from the comultiplication
restricted on this tensor product,
\begin{align*}
  &\Delta(E)= 0,\quad
  \Delta(F)= 0,\quad
  \Delta(K)= 1\otimes 1,\\
  &\Delta(e)= 1\tensor e+e\tensor 1,\quad
  \Delta(f)= 1\tensor f+f\tensor 1,\quad
  \Delta(h)= 1\tensor h+h\tensor 1,
\end{align*}
which coincides with the usual comultiplication for the $s\ell(2)$.
\end{proof}

\subsection{Fusion with projectives}
The above results allow us to decompose tensor products of irreducible
modules with projective ones and of two projective modules. We note
that in both cases the tensor product must contain only projectives.

To decompose the tensor product of an irreducible module and a
projective one, we start with the $r=1$ case as well. We consider
irreducible subquotients and submodules of the projective module and
calculate their tensor products with the irreducible
module. Projectives obtained from these tensor products are direct
summands because any projective $\LUresSL2$-module is also injective
(the contragredient one to a projective module) and is therefore a
direct summand in any module into which it is embedded. Irreducibles
obtained from the tensor products are subquotients of projective
modules in the whole tensor product. This procedure thus gives a
decomposition of the tensor product of an irreducible module with a
projective module.

\begin{lemma}\label{lem-fus-proj}
For $1\le s_1\le p-1$, we have
\begin{equation}\label{eq-fus-proj}
  \repX^{\alpha}_{s_1,1}\tensor\mathscr{P}^{\beta}_{s_2,1} =
  \bigoplus_{\substack{s=|s_1-s_2|+1\\\step=2}}^{\substack{\min(s_1+s_2-1,\\2p-s_1-s_2-1)}}
  \mathscr{P}^{\alpha\beta}_{s,1} +
  2\bigoplus_{\substack{s=2p-s_1-s_2+1\\\step=2}}^{p-\gamma_2}
  \mathscr{P}^{\alpha\beta}_{s,1}
  +\bigoplus_{\substack{s=p-s_1+s_2+1\\\step=2}}^{p-\gamma_1}\mathscr{P}^{-\alpha
  \beta}_{s,2},\\
\end{equation}
where $\gamma_1=(s_1+s_2+1)\mod2$, and $\gamma_2=(s_1+s_2+p+1)\mod2$.
\end{lemma}
\begin{proof}
The projective module $\mathscr{P}^{\beta}_{s_2,1}$ contains three
irreducible subquotients. Let denote them according to their position
in the diagram \eqref{schem-proj-sm}, for simplicity. The top
${\XX^{\beta}_{s_2,1}}$ is marked as $T$, the bottom as $B$, and the
middle ${\XX^{-\beta}_{p-s_2,2}}$ as $M$. The tensor product of the
irreducible module $\repX^{\alpha}_{s_1,1}$ with each of them contains
irreducible terms which can uniquely combain into a projective
module. These projectives are the first sum in the right-hand side of
\eqref{eq-fus-proj}. The tensor product $\repX^{\alpha}_{s_1,1}\otimes
M$ for $s_1>s_2$ gives the
$\bigoplus_{s=p-s_1+s_2+1}^{p-\gamma_1}\mathscr{P}^{-\alpha \beta}_{s,2}$
terms in \eqref{eq-fus-proj}. The tensor products
$\repX^{\alpha}_{s_1,1}\otimes T$ and $\repX^{\alpha}_{s_1,1}\otimes
B$ give
$\bigoplus_{s=2p-s_1-s_2+1}^{p-\gamma_2}\mathscr{P}^{\alpha\beta}_{s,1}$
terms in \eqref{eq-fus-proj}, when $s_1>p-s_2$.
\end{proof}

We also decompose tensor products of two projective modules for the
$r=1$ case analogously as in Lem.~\bref{lem-fus-proj}.
\begin{multline*}
\PP^{\alpha}_{s_1,1}\tensor \PP^{\beta}_{s_2,1}
= 2\bigoplus_{\substack{s=|s_1-s_2|+1\\\text{step}=2}}^{\substack{
\min(s_1 + s_2 - 1,\\ 2p - s_1 - s_2 - 1)}}\PP^{\alpha\beta}_{s,1}
+\bigoplus_{\substack{s=|p-s_1-s_2|+1\\\text{step}=2}}^{\substack{
\min(p-s_1 + s_2 - 1,\\ p + s_1 - s_2 - 1)}}\PP^{-\alpha\beta}_{s,2}
+2\!\!\!\bigoplus_{\substack{s=\min(p-s_1 + s_2 + 1,\\ p + s_1 - s_2 + 1),\;\text{step}=2}}^
{p-\gamma_1}\!\!\!\!\!\!\PP^{-\alpha\beta}_{s,2}\\
+4\bigoplus_{\substack{s=2p-s_1 - s_2 + 1\\\text{step}=2}}^{p-\gamma_2}\PP^{\alpha\beta}_{s,1}
+\bigoplus_{\substack{s=s_1 + s_2 + 1\\\text{step}=2}}^{p-\gamma_2}\bigl(\PP^{\alpha\beta}_{s,1} +
\PP^{\alpha\beta}_{s,3}\bigr).
\end{multline*}

Using Lem.~\bref{fus-lemma}, we now extend these results for arbitrary
$r$-index and obtain the final results in Thm.~\bref{thm-main}.

\subsection{Relations to Virasoro fusion algebra}
In (\ref{identification}), we identify $\LUresSL2$ irreducible
and projective modules with irreducible and logarithmic modules
of the Virasoro algebra $\Vir$. Under the identification,
the tensor products of $\LUresSL2$-modules coincide with the fusion
of the corresponding {[GaberdielKausch]}-modules. In other words, there exists 
a tensor functor from the category $\catC^+_p$ to the category
of $\Vir$-modules with dimension of $L_0$ Jordan cells not greater than 2. 
The functor establishes
a one to one correspondence between simple objects of two categories
but is not an equivalence because the Virasoro category contains
more morphisms between simple objects and more indecomposable
objects than~$\catC^+_p$. In particular, Virasoro Verma modules have no
counterpart on the quantum group side. $\Vir$ also admits a class of modules
with two dimensional $L_0$ Jordan cells enumerated by a projective parameter.
All these modules have the same subquotient structure~(\ref{schem-proj})
nevertheless are parawise different and only module with a special value 
of the parameter has a counterpart on the quantum group side 
(see Rem.~\ref{rem:proj-param}). The details of the correspondence between $\Vir$
and $\LUresSL2$ indecomposable modules will be written in the future paper~\cite{[BFGT]}.




\section{Conclusions}
In the paper, we identified the fusion~\cite{RPfus} of the $\LM(1,p)$ logarithmic
models with the tensor products of $\LUresSL2$ irreducible and projective
modules. The Virasoro chiral algebra $\Vir$ of $\LM(1,p)$ and $\LUresSL2$
centralizes each other in the free-field space of states. This 
suggests that $\Vir$ and $\LUresSL2$ should be in the Kazhdan--Lusztig duality.
However, this duality is more subtle than duality between
$\algW$ and $\UresSL2$ in~\cite{[FGST]}. In the $\Vir$--$\LUresSL2$
case there is no equivalence between category $\catC^+_p$ and
naive category of $\Vir$ representations. An identification
of a relevant $\Vir$ category is an important future problem.

The quantum group $\LUresSL2$ is the maximal centralizer of the
Virasoro algebra $\Vir$ on the full chiral space of states. The two
commuting actions on the space of states are combined into a bimodule
over the Virasoro algebra and the quantum group. This bimodule is
expected to be closely related to the regular bimodule for the
corresponding quantum group, which will be explicitly constructed in
our future paper. The bimodule is resemble the limiting
bimodule~\cite{RS-Q}
that corresponds to the vacuum sector in the XXZ model chiral space 
of states.

In view of the proposed duality between the Virasoro algebra $\Vir$
and the quantum group $\LUresSL2$, there is a correspondence between
representation (sub)categories of the quantum group and the Virasoro
algebra. This allows constructing infinite series of indecomposable
representations for the Virasoro algebra (with the well-known
Feigin--Fuchs modules included) and the Felder resolutions in
quantum-group terms. 

The fusion algebra for the irreducible and indecomposable Virasoro
representations was calculated using the comultiplication in the
quantum group $\LUresSL2$. These results are in good correspondence
with the recent results of Read and Saleur~\cite{RS-Q} based on a
detailed study of quantum XXZ spin chains with the anisotropy
parameter related to the deformation parameter of the quantum group
being a primitive root of unity. This opens a possibility to investigate
XXZ spin chains in terms of $\LUresSL2$. In particular, coefficients
in front of $z$ in characters of multiplicity spaces calculated in~\cite{[FT]}
are related to multiplicities of $\LUresSL2$ projective modules in tensor
products of it's irreducibles and conjecturally give characters of 
Temperley--Lieb algebra modules realized in XXZ spin chains.

\subsubsection*{Acknowledgments} 
We are grateful to M.~Flohr, P.~Pearce, J.~Rasmussen, H.~Saleur 
and A.M.~Semikhatov for valuable
discussions and careful reading of the text.  
The work of BLF was supported in part by RFBR Grant
08-01-00720, RFBR-CNRS-07-01-92214 and LSS-3472.2008.2.  The work of
AMG was supported in part by RFBR Grant~07-01-00523, by the
Grant~LSS-1615.2008.2, and by the ``Landau'', ``Dynasty'' and
``Science Support'' foundations. The work of PVB was supported in part
by RFBR Grant~07-01-00523. The work of IYuT was supported in part by
LSS-1615.2008.2, the RFBR Grant 08-02-01118 and the ``Dynasty''
foundation.

\appendix

\section{Central idempotents}\label{app:idem}
We recall the central idempotents in $\UresSL2$~\cite{[FGST]}:
\begin{multline*}
  \idem_s=\ffrac{1}{\psi_s(\beta_s)}\bigl(
  \psi_s(\cas)-\frac{\psi_s^\prime(\beta_s)}{\psi_s(\beta_s)}(\cas
   -\beta_s)\psi_s(\cas)\bigr),\quad
  1\leq s\leq p-1,\\
  \idem_0=\ffrac{1}{\psi_0(\beta_0)}\psi_0(\cas),\qquad
  \idem_p=\ffrac{1}{\psi_p(\beta_p)}\psi_p(\cas),
\end{multline*}
with the polynomials
 \begin{multline*}
   \psi_s(x)=(x-\beta_0)\,(x-\beta_p)
   \smash{\prod_{\substack{j=1\\
         j\neq s}}^{p-1}}(x-\beta_j)^2,\quad 1\leq s\leq p-1,\\
      \psi_0(x)=(x-\beta_p)\prod_{j=1}^{p-1}(x-\beta_j)^2, \quad 
      \psi_p(x)=(x-\beta_0)\prod_{j=1}^{p-1}(x-\beta_j)^2,
\end{multline*}
where $\beta_j=\q^j+\q^{-j}$, 
and the Casimir element
\begin{equation*}
  \cas=(\q-\q^{-1})^2 EF+\q^{-1}K+\q K^{-1}=(\q-\q^{-1})^2 FE+\q K+\q^{-1}K^{-1}.
\end{equation*}

\section{Projective $\LUresSL2$-modules}\label{app:proj-mod-base}
Here, we explicitly describe the $\LUresSL2$ action in the projective
module $\PP^{\pm}_{s,r}$. Let $s$ be an integer $1\leq s\leq p-1$ and
$r\in\oN$.  

For $r > 1$, the projective module $\PP^{\pm}_{s,r}$ has the basis
\begin{equation}\label{left-proj-basis-plus}
  \{\toppr_{n,m},\botpr_{n,m}\}_{\substack{0\le n\le s-1\\0\le m\le r-1}}
  \cup\{\leftpr_{k,l}\}_{\substack{0\le k\le p-s-1\\1\le l\le
  r-1}}
\cup\{\rightpr_{k,l}\}_{\substack{0\le k\le p-s-1\\0\le l\le
  r}},
\end{equation}
where $\{\toppr_{n,m}\}_{\substack{0\le n\le s-1\\0\le m\le
    r-1}}$ is the basis
corresponding to the top module in~\eqref{schem-proj},\\
$\{\botpr_{n,m}\}_{\substack{0\le n\le s-1\\0\le m\le r-1}}$
to the bottom, $\{\leftpr_{k,l}\}_{\substack{0\le k\le
    p-s-1\\1\le l\le r-1}}$ to the left, and
$\{\rightpr_k\}_{\substack{0\le k\le p-s-1\\0\le l\le r}}$ to
the right module. 

For $r=1$, the basis does not contain
$\{\leftpr_{k,l}\}_{\substack{0\le k\le p-s-1\\1\le l\le
r-1}}$ terms and we imply $\leftpr_{k,l}~\equiv~0$ in the
action.  The $\LUresSL2$-action on $\PP^{\pm}_{s,r}$ is given by
\begin{align*}
  K\toppr_{n,m}&=\pm\q^{s-1-2n}\toppr_{n,m}, \quad 0\le n\le s-1,\quad 0\le m\le r-1,\\
  K\leftpr_{k,m}&=\mp\q^{p-s-1-2k}\leftpr_{k,m}, \quad 0\le k\le p-s-1,\quad 1\le m\le r-1,\\
  K\rightpr_{k,m}&=\mp\q^{p-s-1-2k}\rightpr_{k,m}, \quad 0\le k\le p-s-1,\quad 0\le m\le r,\\
  K\botpr_{n,m}&=\pm\q^{s-1-2n}\botpr_{n,m}, \quad 0\le n\le s-1,\quad 0\le m\le r-1,\\
  E\toppr_{n,m}&=
  \begin{cases}
    \pm[n][s-n]\toppr_{n-1,m}\pm g\botpr_{n-1,m}, &1\le n\le s-1,\\
    \pm g\frac{r-m}{r}\rightpr_{p-s-1,m}\pm g\frac{m}{r}\leftpr_{p-s-1,m}, & n=0,\\
  \end{cases}
  \quad 0\le m\le r-1,\\
  E\leftpr_{k,m}&=
  \begin{cases}
	\mp[k][p-s-k]\leftpr_{k-1,m}, &1\le k\le p-s-1, \\
    \pm g(m-r)\botpr_{s-1,m-1}, & k=0,\\
  \end{cases}
  \quad 1\le m\le r-1,\\
  E\rightpr_{k,m}&=
  \begin{cases}
    \mp[k][p-s-k]\rightpr_{k-1,m}, &1\le k\le p-s-1,\\
    \pm gm\botpr_{s-1,m-1}, & k=0,\\
  \end{cases}
  \quad 0\le m\le r,\\
  E\botpr_{n,m}&= \pm[n][s-n]\botpr_{n-1,m}, \quad 1\le n\le
  s-1, \quad 0\le m\le r-1 \quad (\botpr_{-1,m}\equiv 0),\\
  F\toppr_{n,m}&=
  \begin{cases}
    \toppr_{n+1,m}, &0\le n\le s-2,\\
    \frac{1}{r}\rightpr_{0,m+1}-\frac{1}{r}\leftpr_{0,m+1}, & n=s-1 
    \quad(\leftpr_{0,r}\equiv0),\\
  \end{cases}
   \quad 0\le m\le r-1,\\
  F\leftpr_{k,m}&=
  \begin{cases}
    \leftpr_{k+1,m}, &0\le k\le p-s-2,\\
    \botpr_{0,m}, & k=p-s-1,\\
  \end{cases}
  \quad 1\le m\le r-1,\\
  F\rightpr_{k,m}&=
  \begin{cases}
    \rightpr_{k+1,m}, & 0\le k\le p-s-2,\\
    \botpr_{0,m}, & k=p-s-1,\\
  \end{cases}
  \quad 0\le m\le r,\\
  F\botpr_{n,m}&= \botpr_{n+1,m}, \quad 1\le n\le s-1,
  \quad 0\le m\le r-1 \quad (\botpr_{s,m}\equiv 0).
\end{align*}
where $g=\frac{(-1)^{p}[s]}{[p-1]!}$.

In thus introduced basis, the $s\ell(2)$-generators $e$, $f$ and $h$
act in $\PP^{\pm}_{s,r}$ as in the direct sum
$\XX^\alpha_{s,r}\oplus\XX^{-\alpha}_{p-s,r-1}\oplus\XX^{-\alpha}_{p-s,r+1}
\oplus \XX^\alpha_{s,r}$
(see~\eqref{basis-lusz-irrep-1}-\eqref{basis-lusz-irrep-3}), where for
$r=1$ we set $\XX^{-\alpha}_{p-s,0}\equiv 0$.

\end{document}